\documentclass[a4paper,USenglish,cleveref, autoref, thm-restate]{lipics-v2021}

\pdfoutput=1

\nolinenumbers
\title{Near-Optimal Resilient Labeling Schemes}

\author{Keren Censor-Hillel}{Department of Computer Science, Technion, Israel \and \url{http://ckeren.net.technion.ac.il/}}{ckeren@cs.technion.ac.il}{https://orcid.org/0000-0003-4395-5205}{The research is supported in part by the Israel Science Foundation (grant 529/23).}

\author{Einav Huberman}{Department of Computer Science, Technion, Israel}{einav.hu@cs.technion.ac.il}{https://orcid.org/0009-0005-9675-7669}{}

\authorrunning{K. Censor-Hillel and E. Huberman}

\Copyright{Keren Censor-Hillel and Einav Huberman}



\category{}

\relatedversion{}

\usepackage[utf8]{inputenc}
\usepackage[linesnumbered,noresetcount,vlined,ruled,procnumbered]{algorithm2e}
\usepackage{algpseudocode}
\usepackage{amsfonts}
\usepackage{graphicx}
\usepackage{amsthm}
\usepackage{mathrsfs}
\usepackage{hyperref}
\usepackage{amsmath}
\usepackage{cite}
\usepackage{xcolor}
\usepackage{amssymb}
\usepackage{mathtools}
\usepackage{comment}

\let\oldnl\nl
\newcommand{\nonl}{\renewcommand{\nl}{\let\nl\oldnl}}

\usepackage{caption}

\newcommand{\floor}[1]{\left\lfloor #1 \right\rfloor}
\newcommand{\ceil}[1]{\left\lceil #1 \right\rceil}
\newcommand{\congest}{\ensuremath{\mathsf{Congest~}}}

\newcommand{\F}{\ensuremath{F}}
\newcommand{\A}{\ensuremath{\alpha}}
\newcommand{\B}{\ensuremath{\beta}}
\newcommand{\ComputeGroups}{{\tt{ComputeGroups}}}
\newcommand{\RelayGroups}{{\tt{RelayGroups}}}

\EventEditors{Silvia Bonomi, Letterio Galletta, Etienne Rivi\`{e}re, and Valerio Schiavoni}
\EventNoEds{4}
\EventLongTitle{28th International Conference on Principles of Distributed Systems (OPODIS 2024)}
\EventShortTitle{OPODIS 2024}
\EventAcronym{OPODIS}
\EventYear{2024}
\EventDate{December 11--13, 2024}
\EventLocation{Lucca, Italy}
\EventLogo{}
\SeriesVolume{324}
\ArticleNo{9}

\begin{document}

\maketitle

\begin{abstract}
Labeling schemes are a prevalent paradigm in various computing settings. In such schemes, an oracle is given an input graph and produces a label for each of its nodes, enabling the labels to be used for various tasks. Fundamental examples in distributed settings include distance labeling schemes, proof labeling schemes, advice schemes, and more.
This paper addresses the question of what happens in a labeling scheme if some labels are erased, e.g., due to communication loss with the oracle or hardware errors.
We adapt the notion of resilient proof-labeling schemes of Fischer, Oshman, Shamir [OPODIS 2021] and consider resiliency in general labeling schemes. A resilient labeling scheme consists of two parts -- a transformation of any given labeling to a new one, executed by the oracle, and a distributed algorithm in which the nodes can restore their original labels given the new ones, despite some label erasures.

Our contribution is a resilient labeling scheme that can handle $\F$ such erasures. Given a labeling of $\ell$ bits per node, it produces new labels with multiplicative and additive overheads of $O(1)$ and $O(\log(\F))$, respectively. The running time of the distributed reconstruction algorithm is $O(\F+(\ell\cdot \F)/\log{n})$ in the \textsf{Congest} model. 

This improves upon what can be deduced from the work of Bick, Kol, and Oshman [SODA 2022], for non-constant values of $F$. 
It is not hard to show that the running time of our distributed algorithm is optimal, making our construction near-optimal, up to the additive overhead in the label size.
\end{abstract}

\section{Introduction}
Assigning labels to nodes of a graph is a part of many fundamental computing paradigms. In distributed computing, cornerstone examples include distance and connectivity labeling schemes, as well as labeling schemes for nearest common ancestor \cite{ParterP22a, peleg2000proximity,gavoille2003compact,gavoille2004distance,bar2022fault,abraham2012fully,korman2010labeling,katz2004labeling,izsak2012note,parter2023connectivity,izumi2023deterministic,dory2021fault,gawrychowski2018labeling, alstrup2002nearest}, proof labeling schemes and distributed proofs \cite{kol2018interactive,feuilloley2021redundancy,KormanKP10,feuilloley2021introduction,emek2022locally,goos2011locally,goos2016locally,korman2006distributed,korman2006distributedA,EmekG20,Censor-HillelPP20}, advice schemes \cite{fraigniaud2009distributed, fraigniaud2007local,fraigniaud2010communication,fusco2008trade}, and more, where the above are only samples of the extensive literature on these topics. The common theme in all of the above is that an oracle assigns labels to nodes in a graph, which are then used by the nodes in a distributed algorithm. 
Since distributed networks are prone to various types of failures, Fischer, Oshman and Shamir~\cite{fischer2022explicit} posed the natural question of how to cope with nodes whose labels have been erased and defined the notion of \emph{resilient proof labeling schemes}.

We will use the following definition, 
which applies to a labeling scheme for \emph{any} purpose but captures the same essence.
A resilient labeling scheme consists of two parts. The first part is an oracle transformation of an input labeling $\varphi: V \rightarrow \{0,1\}^{\ell}$ into an output labeling $\psi: V \rightarrow \{0,1\}^{\ell'}$, where $\ell'=A\ell+B$. The second part is a distributed \congest\footnote{In the \congest model, $n$ nodes of a graph communicate in  synchronous rounds by exchanging $O(\log{n})$-bit messages along the edges of the graph.} algorithm, in which each node $v$ receives $\psi(v)$ as its input, except for at most $\F$ nodes where $\psi(v)$ is erased, and each node must recover $\varphi(v)$. The parameters of interest are (i) the number $\F$ of allowed label erasures, (ii) the multiplicative overhead $A$ and the additive overhead $B$ in the label sizes, and (iii) the round complexity of the distributed algorithm.

The aforementioned work of Fischer et al.~\cite{fischer2022explicit} provides a resilient labeling scheme for the case of \emph{uniform} labels, i.e., when the labels of all nodes are the same. For a general case, the work by Bick, Kol, and Oshman~\cite{bick2022distributed} implies  
a labeling scheme that is resilient to $F$ failures, with a constant multiplicative overhead in the label size, within  $O(\F^3 \log \F +(\ell\cdot \F^3)/\log n)$ rounds. 
In this paper, we ask how far we can reduce the overhead costs in the label size and the round complexity of the distributed algorithm.

{\center{\emph{\textbf{Question}: What are the costs of resilient labeling schemes?}}}

~\\When $F$ is constant, the work of Bick et al.~\cite{bick2022distributed} achieves constant overheads in the label size and algorithm complexity. We show that for larger values of $F$, one can do better. The following states our contribution.

\begin{restatable}{theorem}{ThmMain}
\label{thm:main}
There exists a resilient labeling scheme that tolerates $\F$ label erasures, has $O(1)$ and $O(\log{F})$ multiplicative and additive overheads to the size of the labels, respectively, and whose distributed \congest algorithm for recovering the original labels has a complexity of at most $O(\F+(\ell\cdot \F)/\log n)$ rounds.
\end{restatable}

Notice that $\Omega(\F)$ is a straightforward lower bound for the complexity of such a distributed algorithm. To see this, suppose the network is a path. If the labels of $\F$ consecutive nodes along the path get erased, then for the middle node $v$ to obtain any information about its label, it must receive information from some node whose label is not erased. A message from such a node can only reach $v$ after $\Omega(\F)$ rounds. Further, this example shows that $\Omega((\ell\cdot \F)/\log n)$ is also a lower bound for the number of rounds, due to an information-theoretic argument: the nodes on the subpath of $P$ consisting of its middle $\F/2$ nodes must obtain $O(\ell\cdot F)$ bits of information. This number of bits must pass over the two single edges that connect this subpath to the rest of $P$. Since each of the two edges can carry at most $O(\log{n})$ bits, we get that $\Omega((\ell\cdot \F)/\log n)$ rounds are required. This implies that our scheme is nearly optimal, leaving only the $O(\log{F})$ additive overhead in the size of the label as an open question.

\subsection{Technical Overview}
To use a labeling scheme despite erasures, it must be strengthened by backing up information in case it is lost. Erasure-correcting codes are designed for this purpose, but they apply in a different setting -- the entire codeword is available to the computing device, except for some bounded number of erasures. Yet, to handle label erasures, we need to encode the information and split the codeword among several nodes.
Thus, when labels in a labeling scheme are erased, restoring them requires faulty nodes to obtain information from other nodes. This means that the oracle that assigns labels must encode its original labels into new, longer ones, and the  nodes must communicate to restore their possibly erased labels. 

~\\\textbf{Partitioning:} By the above discussion on using error-correction codes, it is clear that if we partition the nodes of the graph into groups of size $\Theta(\F)$, we can encode the labels within each such subgraph, as follows. By ensuring that each group forms a connected subgraph, the nodes can quickly collect all non-erased labels within a group in $O(F+(\ell \cdot F)/\log{n})$ rounds, and thus decode and reconstruct all original labels. 

When $\F$ is large, simple methods can construct such a partition with negligible complexity overhead compared to $O(\F)$. However, when $\F$ is small, such overheads become significant. While this motivation is clear for the complexity of the distributed algorithm, it also applies to the overhead in the label size: much work is invested in providing labels of small size, e.g., the MST advice construction of Fraigniaud, Korman, and Lebhar \cite{fraigniaud2007local}.

The caveat is that if $\F$ is small, there may not exist a distributed algorithm for such a partitioning that completes within $O(\F)$ rounds. To see why, consider a simple reduction from $3$-coloring a ring, which is known to have a lower bound of $\Omega(\log^*{n})$ by Linial \cite{linial1992locality}. Given a partition into groups (paths) of size at least $3$, we can color each group node-by-node in parallel, except for its two endpoints. This takes at most $O(\F)$ rounds, as this is the size of each group. In a constant number of additional rounds, we can in parallel stitch all groups by coloring every pair of adjacent endpoints in a manner consistent with their other two neighbors, which are already colored. Thus, an $O(\F)$-round algorithm for obtaining such a partition contradicts the lower bound for small values of $\F$.

We stress that it is essential that the partition obtained by such a partitioning algorithm is exactly the same one used by the oracle for encoding. In particular, this rules out randomized partitioning algorithms or algorithms whose obtained partition depends on the identity of nodes with label erasures.

Our approach to address this issue is to have the oracle include partitioning information in the new labeling $\psi$. This way, even if $O(\F)$ rounds are insufficient to \emph{construct the partition from scratch}, the nodes can still \emph{reconstruct the partition using the labels} within $O(\F)$ rounds, despite any $\F$ label erasures. A straightforward method would be to include all the node IDs of a group in each label. However, this would increase the size $\ell'$ of the new labeling $\psi$ to at least $O(\F\log{n})$, which for small values of $\F$ is an overhead we aim to avoid. Thus, we need a way to hint at the partition without explicitly listing it. 

~\\\textbf{Ruling sets:} To this end, we employ ruling sets. An $(\A,\B)$-ruling set is a set $S\subseteq V$ such that the distance between every two nodes in $S$ is at least $\A$, and every node is within distance at most $\B$ to the set $S$. Given a ruling set with values of $\A,\B$ that are $O(\F)$, we prove that $O(\F)$ rounds are sufficient for obtaining a good partition.

Interestingly, the properties of the partition we obtain are slightly weaker but still sufficient for our needs (and still rule out an $O(\F)$-round algorithm for constructing them from scratch, by a similar reduction from $3$-coloring a ring). The partition we arrive at is a partition into groups of $\Theta(\F)$ nodes, where each group may not be connected but has low-congestion shortcuts of diameter $O(\F)$ and congestion $2$. This will allow us to collect all labels within a group in $O(\F+(\ell \cdot F)/\log{n})$ rounds.

More concretely, we aim to construct a $(2f+2,2f+1)$-ruling set $S$, for some parameter $f\geq \F$. For intuition, assume $f=\F$. We create groups by constructing a BFS tree $T_v$ rooted at each node $v\in S$, capped at a distance of $2f+1$ (with ties broken according to the smaller root ID). We then communicate only over tree edges to divide the nodes of $T_v$ into groups.
We choose $\alpha=2f+2$ as a lower bound for the distance between two nodes in $S$, ensuring that each root $v$ has at least $f+1$ nodes in its tree. These nodes are within distance $f+1$ from $v$ and cannot be within that distance from any other node in the ruling set due to our choice of $\alpha=2f+2$. 
Additionally, to ensure that the round complexity does not exceed $O(f)$, we need to bound the diameter of each tree. Any $\beta=O(f)$ would suffice. However, as we will argue, computing the ruling set in the distributed algorithm of our resilient labeling scheme is too expensive. Thus, the oracle must provide information in its new labels to help nodes reconstruct the ruling set $S$. We choose $\beta = 2f+1$ so that a node $u$ at distance $2f+2$ from a vertex $v \in S$ can determine it must be in $S$ even if its label is erased.

As promised, we note that due to the lower bound of Balliu, Brandt, Kuhn, and Olivetti \cite{DBLP:conf/stoc/Balliu0KO22}, no distributed algorithm can compute even the simpler task of a $(2,\beta)$-ruling set in fewer than $\Omega(\log{n}/\beta\log\log{n})$ rounds.  
Thus, in an algorithm that takes $O(\F)$ rounds to construct a ruling set with $\beta=\Theta(\F)$,  $\F$ must be at least $\Omega((\log n/\log\log n)^{1/2})$. This leaves the range $\F = O((\log n/\log\log n)^{1/2})$ as an intriguing area. We emphasize that such label sizes are addressed by substantial research on labeling schemes, e.g.:
Korman, Kutten, and Peleg \cite{KormanKP10} assert that for every value of $\F$, there exists a proof-labeling scheme with this label length. Baruch, Fraigniaud, and  Patt-Shamir \cite{baruch2015randomized} prove that a proof-labeling scheme of a given length can be converted into a randomized version with logarithmic length. Patt-Shamir and Perry \cite{patt2017proof} provide additional proof-labeling schemes with small non-constant labels, e.g., they prove the existence of a proof-labeling scheme for the $k$-maximum flow problem with a label size of $O(\log k)$.\footnote{To our knowledge, the fastest deterministic construction in \congest takes $\tilde{O}(\log^3 n)$ rounds, due to Faour, Ghaffari, Grunau, Kuhn, and V{\'a}clav \cite{faour2023local}. Therefore, the actual complexity may be even larger, which would further widen the range of $\F$ in which our construction outperforms other methods.}

The existence of this lower bound could be overcome by adding a single bit to each node's new label to indicate whether the node is in the ruling set. However, another significant obstacle is that the nodes need to reconstruct the ruling set given $\F$ erasures. This turns out to be a major issue, as follows.

Suppose we have a path $(w_0,\cdots, w_{f+1}, \cdots, w_{2f+1})$, with two nodes $u,v$ connected to $w_{2f+1}$, and the rest of the graph is connected through $w_0$. Suppose a $(2f+2,2f+1)$-ruling set $S$ contains $w_0$ and one of $u$ or $v$, and that $ID(v)<ID(w_0)<ID(u)$. If the labels of $u$ and $v$ both get erased, then the node $w_{f+1}$ cannot determine whether it should be part of the tree $T_{w_0}$ or not. Its distance from $v,u$ and $w_0$ is $f+1$, but the ID comparisons for $w_0,u$ and $w_0,v$ yield different results: if $u\in S$, then $w_{f+1}$ should join $T_{w_0}$, but if $v\in S$, then $w_{f+1}$ should join $T_v$. Notice that all nodes with non-erased labels have the same distance to $u$ as to $v$, making it impossible to distinguish between the two possibilities.

This seems to bring us to a dead-end, as the simple workaround of providing each node with the identity of its closest node in the ruling set $S$ is too costly. It would require adding $\Theta(\log{n})$ bits to the new labeling $\psi$, which we wish to avoid. 

~\\\textbf{Greedy ruling sets:} Our main technical contribution overcomes this challenge by demonstrating that a greedy ruling set 
\emph{does} allow the nodes to recover it given $\F$ label erasures. Such a  ruling set is constructed by iterating over the nodes in increasing order of IDs. When a node $v$ is reached, it is added to the ruling set, and all nodes within distance $2f+1$ from $v$ are excluded from the ruling set and ignored in later iterations.
Intuitively, in our above example, this ensures that $S$ contains $v$ and not $u$. Thus, when their labels are erased, the nodes can deduce that $v$ should be in $S$ because it has a smaller ID. 

While this approach solves the previous example, it is still insufficient. Suppose there is a sequence of nodes $v_1,\dots,v_f$ with decreasing IDs whose labels get erased and are not \emph{covered} by any other node in $S$ (a node $v$ is covered by a node in $S$ if their distance is at most $2f+1$). Assume that the distance between every two consecutive nodes in the sequence is $\Theta(f)$. This means that each node $v_i$ knows it needs to be in the ruling set if and only if $v_{i+1}$ is not, essentially unveiling the decisions of the greedy algorithm. However, this implies that $v_1$ needs to receive information about $v_f$, which may be at distance of $\Theta(f^2)$, implying more rounds than our target running time.

To cope with the latter example, the oracle will provide each node with its distance from $S$. This incurs an additive $O(\log{\F})$ bits in the label size. This information is useful for deciding, for some of the nodes with erased labels, whether they are in $S$ or not and, in particular, it breaks the long chain of the latter example. Yet, notice that in the former example, this information does not help us, since all nodes apart from $u$ and $v$ have the same distance to $u$ as to $v$. The reason this helps in one example but not in the other is as follows. For any two nodes $u$ and $v$ with erased labels that need to decide which of them is in $S$, as long as we have at least $f-1$ additional nodes with a \emph{different} distance from $u$ and $v$, we can easily choose between $u$ and $v$, because in total there can be at most $f$ label erasures, so one of these $f-1$ reveals which node needs to be chosen. 

Thus, to reconstruct a greedy $(2f+2,2f+1)$-ruling set $S$ using the above information in the non-erased labels, the nodes do the following. First, a node that sees a non-erased label in $S$ within distance $2f+1$ from it can deduce that it is not in $S$. Second, a node for which all labels within distance $2f+1$ are not erased and are not in $S$ can deduce that it is in $S$. For a remaining node $u$ with an erased label, the distance from $S$ of the non-erased labels in its $2f+1$ neighborhood helps in deciding whether $u$ is in $S$ or not: if it has a node within this distance whose distance from $u$ differs from its distance from $S$, then $u$ is not in $S$.
At this point in the reconstruction algorithm, the only remaining case is a node $v$ which is not known to be covered by nodes in $S$ at this stage and has a non-empty set of nodes $X(v)$ within distance $f$ from it whose labels are erased. We show that such nodes cannot create distant node sequences as in the latter example, which allows $v$ to simply check if it has the smallest ID in $X(v)$, indicating whether it should be in $S$.

Formally, we refer to such nodes as \emph{alternative nodes}, defined as follows. 
Let $dist(v,u)$ be the distance between nodes $v$ and $u$, and let $B_{t}(v)$ be the set of nodes at distance at most $t$ from $v$. For a node $v$, a node $u\ne v$ is called an \emph{alternative node} if $dist(u,v)\leq f$ and there are at most $f-2$ nodes $q\in B_{f+1}(v)\cup B_{f+1}(u)$ with $dist(q,v)\ne dist(q,u)$. 
We emphasize that proving that the only remaining case of alternative nodes is technical and non-trivial.

Given any $f\geq\F$, we prove that providing each node with a single-bit indication of whether it is in the greedy $(2f+2,2f+1)$-ruling set $S$ and an additional $O(\log{f})$ bits holding its distance from $S$ is sufficient for the nodes to reconstruct $S$ in $O(f)$ rounds, despite $\F$ erasures. Due to the codes we choose, we will need $f$ to be slightly larger than $\F$. However, it is sufficient to choose $f=\Theta(\F)$, so to reconstruct the greedy ruling set, the label size needs an additive overhead of $O(\log{\F})$ and the number of rounds remains $O(\F)$.

~\\\textbf{Wrap-up:} 
We can now describe our resilient labeling scheme in full. 

In Section \ref{section: assigning labels}, we use a greedy ruling set with $f=\Theta(\F)$ to construct groups and apply an error-correcting code to the labels of their nodes. This provides the entire algorithm of the oracle for transforming an $\ell$-labeling $\varphi$ into an $\ell'$-labeling $\psi$, such that $\ell'$ has a constant multiplicative overhead for coding and an additive overhead of $O(\log{\F})$ for restoring the greedy $(2f+2,2f+1)$-ruling set.

In Section \ref{section: restoring a greedy ruling set despite f faults}, we provide a distributed algorithm for restoring the greedy $(2f+2,2f+1)$-ruling set given the labeling $\psi$ despite $\F$ label erasures, within $O(\F)$ rounds. In Section \ref{section: partitioning}, the nodes partition themselves into the same groups as the oracle, within $O(\F)$ rounds. Finally, Section \ref{section: f-resilient labeling schemes} presents our complete resilient labeling scheme, which requires an additional number of $O(\F+(\ell\cdot \F)/\log{n})$ rounds for communicating the non-erased labels within each group and decode to obtain the original labeling $\varphi$.

\subsection{Preliminaries}
\label{ssubection:preliminaries}
Throughout the paper, let $G=(V,E)$ be the network graph, which without loss of generality is considered to be connected. Let $n$ denote the number of nodes in $G$. After labels are assigned to nodes by the oracle, there is a parameter $\F$ of the number of erasures of labels that may occur. For each node $v\in V$, we denote by $ID(v)$ the unique identifier of $v$.

We denote by $dist(v,u)$ the distance between two nodes $v$ and $u$, and by $P_{v,u}$ some shortest path between them.
The ball of radius $t$ centered at a node $v$ consists of all nodes within distance $t$ from $v$ and is denoted by $B_t(v)=\{u\in V \mid dist(u,v)\leq t\}$.

Throughout the paper, we use $a\circ b$ to refer to the concatenation of two strings $a$ and $b$.

\subsection{Related Work}
\label{ssubection:related work}
As mentioned earlier, the work of Bick, Kol, and Oshman \cite{bick2022distributed} implicitly suggests a resilient labeling scheme with an $O(1)$ multiplicative overhead to the label size in $O(\F^3 \log \F +(\ell\cdot \F^3)/\log n)$ rounds. This is achieved by constructing a so-called $(k,d,c,h)$-helper assignment, which assigns $k$ helper nodes to each node $v$ in the graph, with paths of length at most $d$ from each $v$ to all of its helper nodes, such that each edge belongs to at most $c$ paths, and a node is assigned as a helper node to at most $h$ nodes $v$. The paper computes a $(k,O(k),O(k^2),O(k))$-helper assignment in $O(k\log k)$ rounds, using messages of $O(k^2\log n)$ bits. Plugging in $k=\Theta(F)$ gives the aforementioned resilient labeling scheme. For $\F=\omega(1)$, the parameters we obtain improve upon the above.

Ostrovsky, Perry, and Rosenbaum \cite{DBLP:conf/sirocco/OstrovskyPR17} introduced $t$-proof labeling schemes, which generalize the classic definition by allowing $t$ rounds for verification, and study the tradeoff between $t$ and the label size. This line of work was followed up by Feuilloley, Fraigniaud, Horvonen, Paz, and Perry \cite{feuilloley2021redundancy}, and Fischer, Oshman, and Shamir \cite{fischer2022explicit}. All of these papers provide constructions of groups of nodes that allow sharing label information to reduce the label size in this context. However, our construction is necessary for handling: (i) general graphs, (ii) in the \congest model, (iii) with arbitrary labels, and (iv) while achieving our goal complexity. It is likely that our algorithm can be applied to this task as well.
Error-correcting codes have already been embedded in the aforementioned approaches of Bick et al. \cite{bick2022distributed} and Fischer et al. \cite{fischer2022explicit}, and we do not consider their usage a novelty of our result.

Note that Feuilloley and Fraigniaud \cite{FeuilloleyF22} design error-sensitive proof-labeling schemes, whose name may seem to hint at a similar task, but are unrelated to our work: these are proof-labeling schemes where a proof that is far from being correct needs to be detected by many nodes.

\section{Assigning Labels} \label{section: assigning labels}
To allow the nodes in the distributed algorithm to decode their original labels despite $\F$ label erasures, we employ an error-correcting code in our resilient labeling scheme. Known codes can correct a certain number of errors at the cost of only a constant overhead in the number of bits. Therefore, we can split the nodes into groups of size $\Theta(\F)$ and encode the labels within each group. We include in the new labels information about the group partitioning to allow the distributed algorithm to quickly reconstruct the same groups. This information comes in the form of a ruling set, on which we base the group partitioning. To ensure that this reconstruction can be done despite label erasures, we ensure that the oracle uses a greedy ruling set. 
We present the oracle's algorithm and then prove its guarantees in Theorem \ref{theorem:code analysis}.

~\textbf{Overview:} In Algorithm \ref{alg: assigning labels}, the oracle first constructs a greedy $(2f+2,2f+1)$-ruling set $S$. For every node $v$, let $b(v)$ be a single bit indicating whether $v\in S$, and let $distS(v)$ be an $O(\log{f})$-bit value which is the distance between $v$ and $S$.

The oracle then partitions the nodes into groups of size $\Theta(f)$. To later allow the distributed algorithm to reconstruct the same groups, the oracle uses the same algorithm the nodes will later use, i.e., Algorithm \ref{alg: partitioning}. However, for the description in this section, any partitioning with the same guarantees would give the desired result in Theorem \ref{theorem:code analysis} below. 

Finally, the oracle applies an error-correcting code to each index of the labels of all nodes in each group. The new label it assigns to each node $v$ is a concatenation of its $b(v)$ and $distS(v)$ values, along with the encoded label.

\begin{algorithm}[htbp]
\caption{Oracle's algorithm for assigning new labels}\label{alg: assigning labels}
\setcounter{AlgoLine}{0}
\KwIn{A graph $G$, a parameter $\F$, and an $\ell$-labeling $\varphi:V\rightarrow \{0,1\}^\ell$.}
\KwOut{An $\ell '$-labeling $\psi:V\rightarrow \{0,1\}^{\ell '}$.}
\BlankLine
    Construct a greedy $(2f+2,2f+1)$-ruling set $S$. For every node $v$, denote by $distS(v)$ its distance to $S$\label{oraclestep:ruling}.
    
    Simulate an execution of Algorithm \ref{alg: partitioning} with $b(v)=1$ for every node $v\in S$ and $b(v)=0$ for every node $v\notin S$. At the end of the algorithm, every node $v$ is associated with a group identifier $Group(v)$\label{oraclestep:partition}.
    
    For every group $Group$ obtained in the previous step, let $(v_1,\dots,v_s)$ be the IDs of its nodes in increasing order. For every $1\leq j \leq \ell$, concatenate the $j$-th bit of the labels $\varphi(v)$ of the nodes according to their increasing to obtain a string ${\tt{r}}^j$, and encode this string into a codeword ${\tt{w}}^j$ according to the code ${\tt{C}}$. Split the bits of ${\tt{w}}^j$ into $s$ consecutive \emph{blocks} $({\tt{w}}^j_1,\dots,{\tt{w}}^j_s)$ of sizes $\ceil{|{\tt{w}}^j|/s}$ or $\floor{|{\tt{w}}^j|/s}$. For every $1\leq i \leq s$, let ${\tt{w}}_i={\tt{w}}^1_i,\dots,{\tt{w}}^{\ell}_i$ and set $\psi(v_i)=b(v_i)\circ distS(v_i)\circ {\tt{w}}_i$\label{oraclestep:encode}.
\end{algorithm}

\textbf{Choice of code:} 
Different codes ${\tt{C}}$ result in different overheads of $\ell'$ compared to $\ell$. A simple repetition code takes the string ${\tt{info}}$ of length $s$ and repeats it for $s$ blocks to create the resulting codeword ${\tt{w}}={\tt{C}}({\tt{info}})$, which is trivially decodable given any $s-1$ blocks erasures. Splitting ${\tt{w}}$ back to $s$ pieces implies that ${\tt{w}}_i$ equals ${\tt{info}}$ for every node $v_i$ in the group. Taking $f=\F$, 
this results in $\ell'=1+O(\log{f})+s\cdot\ell=O(\log{\F})+s\cdot\ell$. Since the group size $s$ is always $\Theta(f)=\Theta(\F)$, we get a multiplicative overhead of $O(\F)$ and an additive overhead of $O(\log{\F})$.
To reduce the multiplicative overhead to $O(1)$, we use a code with better parameters. Formally, a binary code ${\tt{C}}$ maps strings of $k$ bits into strings of $N$ bits called \emph{codewords}. The ratio $\rho=N/k$ is the code's \emph{rate}. The \emph{relative distance}, $\delta$, ensures the Hamming distance (number of indices where the strings differ) between any two codewords is at most $\delta N$. Such a code can correct $\delta N - 1$ erasures. 

\begin{lemma}[{Justesen Codes \cite{justesen1972class}, phrasing adopted from \cite{AshkenaziGL22}}]\label{lemma: coding}
For any given rate $\rho < \frac{1}{2}$ 
and for any $M >0$, there exists a binary code ${\tt{C}}_M :\{0, 1\}^{k_M=\rho_M\cdot N_M}\rightarrow\{0, 1\}^{N_M}$ with $N_M=2M(2M-1)$ such that the code ${\tt{C}}_M$ has rate $\rho_M\geq\rho$ and relative distance $\delta_M >(1-2\rho)H^{-1}(1/2)$, where $H(x)=x\log (1/x)+(1-x)\log (1/(1-x))$ is the binary entropy function. Encoding and decoding can be done in polynomial time.
\end{lemma}

Given the bound $\F$ on the number of label erasures, we need to choose a value of $f$ to work with. Then, given a group of size $s$ such that $f+1\leq s\leq 3f+1$, we need to choose values for $\rho$ and $M$ that determine the code ${\tt{C}}_M$. We set these parameters as follows. We let $f=80\F$, and for every such $s$ we choose $\rho=1/4$. To choose $M$, notice that Lemma \ref{lemma: coding} gives that for any $M>0$, the number of bits we can encode is $k_M=\rho_M N_M$. We choose the smallest $M$ such that $k_M$ is at least $s$.

The next theorem proves that the overhead for the number of bits held by each node per coded bit is constant and that, given the coded blocks,it is possible to decode the original string of bits if up to $\F$ blocks are erased. The proof is presented in Appendix \ref{app:proofs}.

\begin{theorem}
\label{theorem:code analysis}
Let $f=80\F$ and let $\rho=1/4$.
Let $(v_1,\dots,v_s)$ be the sequence of nodes in a group of size $s$, such that $f+1\leq s\leq 3f+1$, ordered by increasing IDs. Let ${\tt{r}}$ be a string obtained by concatenating a single bit associated with each $v_i$ according to their order, and let ${\tt{w}}$ be a codeword obtained in Step \ref{oraclestep:encode} of Algorithm \ref{alg: assigning labels}, by using a code ${\tt{C}}_M$ from Lemma \ref{lemma: coding}, with a value of $M$ that is the smallest such that $k_M\geq s$. Then,
\begin{enumerate}
    \item it holds that $\ceil{|{\tt{w}}|/s} \leq 80$, and 
    \item the string ${\tt{r}}$ is decodable given ${\tt{w}} = ({\tt{w}}_1,\dots,{\tt{w}}_s)$ with up to $\F$ erasures of blocks ${\tt{w}}_i$.
\end{enumerate}
\end{theorem}

For later use by the distributed algorithm in our resilient labeling scheme, we denote by $decode(v,\{(u,{\tt{w}}(u)) \mid u \in Group(v)\})$ the function that takes as input a node $v$ and the values ${\tt{w}}(u)$ for all the nodes $u$ in $Group(v)$ (along with their IDs so that they can be ordered correctly) and produces the bit of $v$ in the string ${\tt{r}}$ which was encoded into ${\tt{w}}$, by Theorem \ref{theorem:code analysis}.

\section{Restoring a Greedy Ruling Set Despite Erasures}\label{section: restoring a greedy ruling set despite f faults}
In this section we present a distributed algorithm that restores a greedy ruling set despite label erasures. Nodes are \emph{faulty} if their label are erased, and \emph{non-faulty} otherwise. Each node $v$ aims to restore its value $b(v)$ so that the set $\{v \mid b(v)=1\}$ is exactly the given ruling set $S$. Once a faulty node $v$ sets its value $b(v)$, it is no longer considered faulty for the remainder of the algorithm. 
The number of label erasures our algorithm can tolerate is denoted by $f$, and in a setting with at most $\F$ label erasures, the algorithm can be executed correctly for any value $f\geq \F$.

\begin{theorem}\label{theorem:black or white}
    Let $G = (V, E)$ be a graph and let $S$ be a greedy $(2f+2, 2f+1)$-ruling set $S$. Suppose that, apart from up to at most $f$ faulty nodes, each node $v$ is provided with a label $\zeta(v)=b(v)\circ distS(v)$, where $b(v)=1$ if $v\in S$ and $b(v)=0$ otherwise, and $distS(v)=\min_{u\in S}\{dist(v,u)\}$. Then, there exists a deterministic \congest algorithm at the end of which every node $v$ restores its value $b(v)$. The algorithm runs in $O(f)$ rounds.
\end{theorem}

Before presenting and analyzing our algorithm for Theorem \ref{theorem:black or white}, we recall the definition of alternative nodes.  

\begin{definition} [Alternative Nodes] \label{def: alternative node}
    For a node $v$, a node $u\ne v$ is called an \emph{alternative node}, if $dist(u,v)\leq f$ and there are at most $f-2$ nodes $q\in B_{f+1}(v)\cup B_{f+1}(u)$ with $dist(q,v)\ne dist(q,u)$. Denote by $M(v)$ the set of alternative nodes for $v$.
\end{definition}

A key tool in our analysis is the proof that a greedy ruling set $S$ excludes nodes $v$ with alternative nodes $u\notin S$ having smaller ID. In Appendix \ref{app:proofs}, we prove the following.

\begin{lemma}\label{lemma: greedy}
    Let $S$ be a greedy $(2f+2,2f+1)$-ruling set. Then, there exists no node $v\in S$ for which there is an alternative node $u\notin S$ with $ID(u)<ID(v)$.
\end{lemma}

We are now ready to prove Theorem \ref{theorem:black or white}. We begin by presenting the algorithm. 

\textbf{Overview.} In Algorithm \ref{alg: restoring greedy ruling set}, nodes restore a greedy ruling set $S$ as follows. Before the algorithm starts, each non-faulty node $w$ sets its value $b(w)$ to the first bit of its label. At the beginning of the algorithm, non-faulty nodes in $S$ broadcast a bit ``1" for $2f+1$ rounds. Here, every node that receives a bit ``1" during these rounds forwards it to all of its neighbors. 
Faulty nodes $u$ that receive a bit during these rounds sets $b(u)=0$, as this indicates that $u\in B_{2f+1}(v)$ for some such $v\in S$ and thus $u\notin S$. 

Next, each faulty node propagates its ID through a distance of $2f+1$. Faulty nodes $u$ that do not receive any other ID set $b(u)=1$, as this means no other faulty nodes are in $B_{2f+1}(u)$ and thus $u\in S$. 

Now, every remaining faulty node $u$ propagate a message of the form $ID(u)\circ dist$ through a distance of $f+1$, where $dist$ starts at 0 and is increased by 1 by each node along the way, corresponding to the length of the path. For each node $w$, let $dist_u(w)$ be the smallest $dist$ value received for $u$, representing the actual distance between $w$ and $u$.

Each node $w\in B_{f+1}(u)$ receives the message and checks if it is possible that $u\in S$ by verifying whether $distS(w)=dist_u(w)$. If this equality does not hold, $w$ propagates a message with $ID(u)$ through a distance of $f+1$. Since $dist(w,u)\leq f+1$, if $u\in S$, then $distS(w)=dist_u(w)$. If $w$ propagates $ID(u)$, then this is not the case, so when $u$ receives its own ID, it can conclude that it is not in $S$ and sets $b(u)=0$.

Now, we are left with faulty nodes $v\in S$ and faulty nodes $u \notin S$. We will prove that $u$ is an alternative node for a faulty node $v\in S$. Thus, every faulty node $v\in S$ must have the lowest ID among faulty nodes in $B_{f+1}(v)$ and can deduce that it is in $S$ and set $b(v)=1$. Similarly, every faulty node $u\notin S$ must have a faulty node $v\in S$ within $B_{f+1}(u)$ such that $ID(v)<ID(u)$, allowing $u$ to deduce it is not in $S$ and set $b(u)=0$.

Following is the formal algorithm and proof.

\begin{algorithm}[htbp]
 \caption{Restoring a greedy $(2f+2,2f+1)$-ruling set $S$ despite $f$ label erasures}\label{alg: restoring greedy ruling set}
\setcounter{AlgoLine}{0}
\KwIn{A graph $G$ with a greedy ruling set $S$, where each node $w$ is provided with $\zeta(w)=b(w)\circ distS(w)$, except at most $f$ faulty nodes which are provided with $\zeta(w)=empty$.}
\KwOut{Every node $w$ outputs $b(w)$ such that $\{w \mid b(w)=1\}=S$.}
\BlankLine
    Every non-faulty node $w\in V$ sets $b(w)$ to the first bit in $\zeta (w)$\label{restorestep:1}.
    
    Every node $v\in S$ with $\zeta(v)\ne empty$ broadcasts a bit ``1" for $2f+1$ rounds\label{restorestep:2}.
    
    Every faulty node $u$ that receives a bit ``1" during Step \ref{restorestep:2} sets $b(u)=0$\label{restorestep:3}.
    
    Every faulty node $u$ propagates a message with $ID(u)$ through a distance of $2f+1$. Nodes that have already broadcast $ID(u)$ do not repeat the broadcast\label{restorestep:4}.
    
    Every faulty node $u$ which does not receive any other ID during Step \ref{restorestep:4} sets $b(u)=1$\label{restorestep:5}.
    
    Every faulty node $u$ propagates a message of the form $ID(u) \circ dist$ through a distance of $f+1$, where $dist$ starts at 0 and is increased by 1 by any node along the way. A node does not propagate a message with some ID if it has already done so with a smaller value of $dist$. For each node $w$, let $dist_u(w)$ be the smallest $dist$ value it receives for node $u$\label{restorestep:6}.
    
    Every node $w\in B_{f+1}(u)$ with $\zeta(w)\ne empty$ receives a message during Step \ref{restorestep:6} and checks if $distS(w) = dist_u(w)$. If the equality does not hold, then $w$ propagates a message with $ID(u)$ through a distance of $f+1$. Nodes that have already broadcast $ID(u)$ do not repeat the broadcast\label{restorestep:7}.
    
    Every faulty node $u$ that receives $ID(u)$ during Step \ref{restorestep:7} sets $b(u)=0$.\label{restorestep:8}
    
    Every faulty node $u$ propagates a message with $ID(u)$ through a distance of $f$. \label{restorestep:9}
    
    Let $X(u)$ be the set of nodes whose IDs are received by $u$ during Step \ref{restorestep:9}. If $ID(u) = \min_{u'\in X(u)}\{ID(u')\}$, then $u$ sets $b(u)=1$. Otherwise, it sets $b(u)=0$\label{restorestep:10}.
\end{algorithm}

\begin{proof}[Proof of Theorem \ref{theorem:black or white}]\label{proof of theorem black or white}
    To prove that the algorithm restores the set $S$, we prove that every node $w\in V$ holds $b(w)$ at the end of Algorithm \ref{alg: restoring greedy ruling set}, such that $\{w \mid b(w)=1\}=S$. By Step \ref{restorestep:1}, this holds for all non-faulty nodes, so it remains to prove it for faulty nodes. In what follows, we go over each step in the algorithm where nodes may set their output (Steps \ref{restorestep:3},\ref{restorestep:5},\ref{restorestep:8}, and \ref{restorestep:10}), and we show that any node which sets its output during that step does so correctly.

    \textbf{Step \ref{restorestep:3}.} In Step \ref{restorestep:2}, every non-faulty node $v\in S$ broadcasts a bit ``1" for $2f+1$ rounds. By the definition of a $(2f+2,2f+1)$-ruling set, every two nodes $v',v''\in S$ satisfy that $dist(v',v'')\geq 2f+2$. Thus, if a faulty node $u$ receives a bit ``1" from $v\in S$ during these $2f+1$ rounds, it implies that $u\in B_{f+1}(v)$. Hence, $u\notin S$ since $v\in S$. So, for any node $u$ that sets $b(u)=0$ in Step \ref{restorestep:3}, its output is correct.

    \textbf{Step \ref{restorestep:5}.} Afterwards, in Step \ref{restorestep:4}, each faulty node propagates a message with its ID through a distance of $2f+1$. Every faulty node $u$ that does not receive any other ID during these rounds and did not receive a ``1" during Step \ref{restorestep:2} sets $b(u)=1$, because this implies that there is no other faulty node in $B_{2f+1}(u)$ and hence it must hold that $u\in S$. So, it is correct that $u$ sets $b(u)=1$ in Step \ref{restorestep:5}.

    \textbf{Step \ref{restorestep:8}.} Then, in Step \ref{restorestep:6}, every node $u$ that is still faulty propagates a message of the form $ID(u)\circ dist$ through a distance of $f+1$, where $dist$ starts at $0$ and is increased by $1$ before the message is propagated further. This corresponds to the length of the path that this message traverses. 
    Every node $w\in B_{f+1}(u)$ receives the above message and checks if it is possible that $u\in S$ by verifying whether $distS(w)=dist_u(w)$. We denote by $dist(u)$ the value of $dist$ from the message $ID(u)\circ dist$ that $w$ receives. If equality does not hold, then $w$ propagates a message with $ID(u)$ through a distance of $f+1$. Since $dist(u,w)\leq f+1$, if $u\in S$, then for every other node $v'\in S$ such that $dist(w,v')\leq f+1$, it holds that $2f+2\leq dist(u,v')\leq dist(u,w) +dist(w,v')\leq f+1+dist(w,v')$. Thus, $dist(w,v')\geq f+1$, and since $distS(w)=\min_{v\in S}\{dist(q,v)\}$ we have $distS(w)=dist_u(w)$. But if $w$ propagates $ID(u)$, then this is not the case, so $u\notin S$. Thus, it is correct that $u$ sets $b(u)=0$ in Step \ref{restorestep:8}.
    
    \textbf{Step \ref{restorestep:10}.}
    First we prove that every node $u$ that is still faulty at the beginning of this step and every node $v\in X(u)$ are alternative nodes. By definition of alternative nodes, we need to prove that $dist(u,v)\leq f$ and that $v$ and $u$ have at most $f-2$ nodes $w\in B_{f+1}(v)\cup B_{f+1}(u)$ with $dist(u,w)\ne dist(v,w)$.
    
    First, since $v\in X(u)$, it must be that $dist(u,v)\leq f$. 
    Second, we prove that $v$ and $u$ have at most $f-2$ nodes $t\in B_{f+1}(v)\cup B_{f+1}(u)$ with $dist(u,t)\ne dist(v,t)$. Assume, towards a contradiction, that there are at least $f-1$ nodes $t\in B_{f+1}(v)\cup B_{f+1}(u)$ with $dist(v,t)\ne dist(u,t)$. Since $u$ is still faulty in Step \ref{restorestep:9}, there are at most $f-2$ nodes $q\ne u,v$ in $B_{f+1}(u)$ with $dist(u,q)\ne dist(v,q)$, because otherwise $u$ would receive a message with $ID(u)$ from such a non-faulty node during Step \ref{restorestep:7} and would set $b(u)=0$ in Step \ref{restorestep:8}. Therefore, there exists a node $s\in B_{f+1}(v)\setminus B_{f+1}(u)$ that is still faulty in Step \ref{restorestep:9}.
        We consider two cases, obtaining a contradiction in both, and conclude that the assumption is incorrect. Thus, $v$ and $u$ have at most $f-2$ nodes $t\in B_{f+1}(v)\cup B_{f+1}(u)$ with $dist(u,t)\ne dist(v,t)$.
        
        \textbf{Case 1: No node $r\ne v$ exists on $P_{v,s}\cap P_{v,u}$.}
        An illustration of this case is presented in Figure \ref{fig: section 4 theorem case 1}.
        \begin{figure}[htbp]
            \centering
            \includegraphics[width=0.4\linewidth]{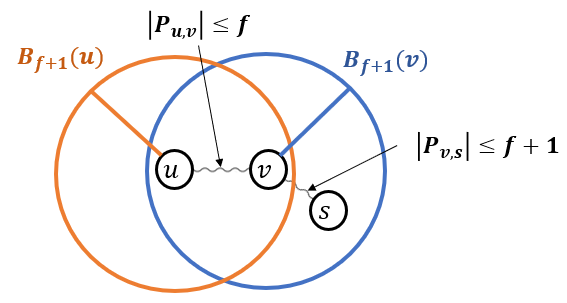}
            \caption{Illustration for Case 1. No node $r\ne v$ exists on $P_{v,s}\cap P_{v,u}$. Therefore, $P_{v,u}$ and $P_{v,s}$ are disjoint except at the node $v$. The orange circle represents $B_{f+1}(u)$ and the blue circle represents $B_{f+1}(v)$. We assume that $dist(u,v)\leq f$, so $|P_{u,v}|\leq f$. Additionally, $s\in B_{f+1}(v)\setminus B_{f+1}(u)$.}
            \label{fig: section 4 theorem case 1}
        \end{figure}
        In this case, every node $t\ne u,v$ on $P_{v,u}$ satisfies that $dist(u,t)\ne dist(v,t)$, except at most one node $a$ (such a node $a$ exists if $dist(u,v)$ is even and $a$ is exactly in the middle of $P_{v,u}$).
        Additionally, every node $t$ on $P_{v,s}$ satisfies that $dist(t,v)\ne dist(t,u)$, as otherwise $dist(u,s)\leq dist(u,t)+dist(t,s)=dist(v,t)+dist(t,s)=dist(v,s)\leq f+1$, which contradicts the assumption that $dist(u,s)\geq f+2$. 
        Denote by $P'$ the path $P_{u,v}\circ P_{v,s}$ at a distance of at most $f+1$ from $u$. There are $f$ nodes on $P'$ that are neither $u$ nor $v$. Additionally, excluding the node $a$, we have $f-1$ nodes $q$ on $P'$ with $dist(v,q)\ne dist(u,q)$, and they are all in $B_{f+1}(u)$. This contradicts the assumption that there are at most $f-2$ such nodes. Therefore, this case in not possible.
        
        \textbf{Case 2: There exists a node $r\ne u$ on $P_{v,s}\cap P_{v,u}$.} An illustration of this case is presented in Figure \ref{fig: section 4 theorem case 2}.
        \begin{figure}[htbp]
            \centering
            \includegraphics[width=0.5\linewidth]{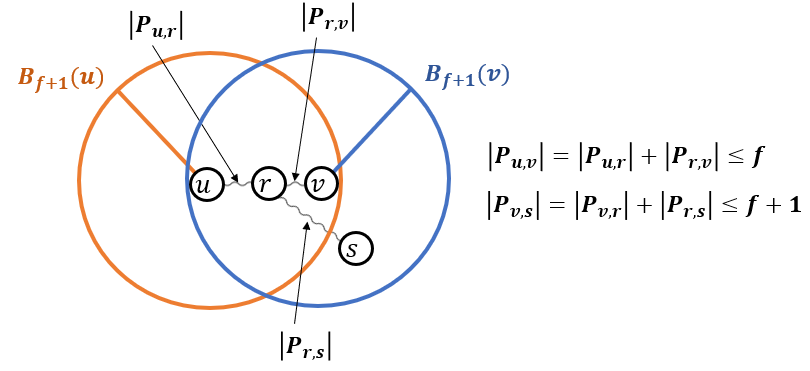}
            \caption{Illustration for Case 2. There exists a node $r\ne v$ on $P_{v,s}\cap P_{v,u}$. We consider $r$ to be the closest node to $u$ in $P_{v,u}\cap P_{v,s}$. The orange circle represents $B_{f+1}(u)$ and the blue circle represents $B_{f+1}(v)$. We assume that $dist(u,v)\leq f$, and since $r$ is on $P_{u,v}$, we have $|P_{u,r}|+|P_{r,v}|\leq f$. Additionally, $s\in B_{f+1}(v)\setminus B_{f+1}(u)$. Since $r$ is on $P_{v,s}$, we have $|P_{v,r}|+|P_{r,s}|\leq f+1$.}
            \label{fig: section 4 theorem case 2}
        \end{figure}
        We consider $r$ to be the closest node to $u$ in $P_{v,u}\cap P_{v,s}$.
        In this case, every node $t\ne u,v$ on $P_{v,u}$ satisfies that $dist(u,t)\ne dist(v,t)$, except at most one node $a$, as we proved in the previous case. Therefore, it is also correct for $P_{u,r}$. 
        Additionally, every node $t$ on $P_{r,s}$ satisfies that $dist(t,v)\ne dist(t,u)$, as otherwise $dist(u,s)\leq dist(u,t)+dist(t,s)=dist(v,t)+dist(t,s)=dist(v,s)\leq f+1$, contradicting the assumption that $dist(u,s)\geq f+2$.

        Denote by $P'$ the path $P_{u,r}\circ P_{r,s}$ at a distance of at most $f+1$ from $u$. There are $f+1$ nodes on $P'$ that are not $u$. Additionally, excluding the node $a$, we have $f$ nodes $q$ on $P'$ with $dist(v,q)\ne dist(u,q)$, and they are all in $B_{f+1}(u)$. This contradicts the assumption that there are at most $f-2$ such nodes. Therefore, this case in not possible.
    
    ~\\
    This completes the proof that every node $u$ that is still faulty at the beginning of this step and every node $v\in X(u)$ are alternative nodes.

    Now, we prove that if $u$ sets $b(u)=0$ in this step, then it satisfies that $u\notin S$. In this case, since $u$ sets $b(u)=0$, there exists a node $v\in X(u)$ with $ID(v)<ID(u)$. Since $v$ is an alternative node for $u$ with $ID(v)<ID(u)$, it follows that $u\notin S$ by Lemma \ref{lemma: greedy}.

    Finally, we prove that every faulty node $u$ that sets $b(u)=1$ satisfies that $u\in S$. By the definition of a $(2f+2,2f+1)$-ruling set, there exists a node $v\in S$ such that $dist(u,v)\leq 2f+1$ (which may be $u$ itself). We claim that $v\in B_{f}(u)$. Assume otherwise, i.e., that $dist(u,v)\geq f+1$.  
    Notice that $v,u$ have at least $|P_{v,u}|-2$ nodes $j$ with $dist(v,j)\ne dist(u,j)$ on $P_{v,u}$, since there are $|P_{v,u}|-1$ nodes $j\ne u,v$ on $P_{v,u}$, and at most one of them has $dist(v,j)=dist(u,j)$ (such a node $j$ exists if $dist(u,v)$ is even and $j$ is exactly in the middle of $P_{v,u}$). Therefore, if $dist(u,v)\geq f+1$, consider nodes on the path $P_{u,v}$ that are at a distance of at most $f+1$ from $u$. There are at least $f-1$ such nodes $j\ne u,v$ with $dist(u,j)\ne dist(v,j)$, and so at least one of them is not faulty, which would lead to $u$ setting $b(u)=0$ in Step \ref{restorestep:8}. But, $u$ is faulty at the beginning of Step \ref{restorestep:9}, which leads to a contradiction.
    
    Thus, $v\in B_{f}(u)$. Additionally, $v\in X(u)$ because in Step \ref{restorestep:9}, $u$ is still faulty, and therefore, in Step \ref{restorestep:3}, $u$ does not set $b(u)=0$, so $v$ must also still be faulty in Step \ref{restorestep:3}. Since $dist(u,v)\leq 2f+1$ and $v$ is still faulty in Step \ref{restorestep:4}, $u$ receives a message with $ID(v)$ in Step \ref{restorestep:4}. Note that if $v$ is not faulty in Step \ref{restorestep:9}, then either it is non-faulty to begin with, in which case $u$ would have already set $b(u)=0$ in Step \ref{restorestep:3}, or $v$ sets $b(v)=1$ in Step \ref{restorestep:5}. The latter is impossible because $dist(u,v)\leq 2f+1$ and $u$ is still faulty in Step \ref{restorestep:9}. Thus, $v$ is still faulty in Step \ref{restorestep:4}, which means that $v$ would receive a message from $u$ in Step \ref{restorestep:4} and not set $b(v)=1$ in Step \ref{restorestep:5}. Thus, $v$ is also still faulty in Step \ref{restorestep:9}, and since $dist(u,v)\leq f$, then $v\in X(u)$.
    
    Thus, since $u$ sets $b(u)=1$, every node $w\in X(u)$ satisfies $ID(u)\leq ID(w)$. Since, $w\in X(u)$, then as we proved, $w$ is an alternative node for $u$. By Lemma \ref{lemma: greedy}, this implies that $w$ cannot be in $S$. Since $v\in X(u)\cap S$, it must hold that $v=u$. Thus $u\in S$, as needed.
    
    This completes the proof that, at the end of Algorithm \ref{alg: restoring greedy ruling set}, we set $b(w)$ for every node $w\in V$ correctly, satisfying $\{w|b(w)=1\}=S$.
    Finally, we sum the number of rounds the algorithm takes. In every broadcast in the algorithm, we have at most $f$ different messages (each of which fits in $O(\log n)$ bits). Moreover, a node passes such a message only once. Thus, each broadcast takes $O(f)$ rounds. Since there are only a constant number of steps with such broadcasts, the algorithm completes in $O(f)$ rounds.
\end{proof}

\section{Partitioning} \label{section: partitioning}

The nodes use the restored ruling set to partition themselves into groups, with each group containing an error correction code to recover the erased labels. The algorithm they follow is the same as that of the oracle, ensuring determinism and independence from faulty nodes. This guarantees consistent partitioning and correct decoding.

To ensure fast communication within each group, our goal is to partition the graph nodes into groups of size $\Theta (f)$. Our algorithm achieves guarantees that are slightly weaker, in the sense that a group may be a disconnected component in the graph, but its weak diameter (its diameter in the original graph) is still $O(f)$. In general, such a guarantee may still be insufficient for fast communication within each group, as some edges may be used by multiple groups, causing congestion. However, our algorithm will make sure to provide low-congestion shortcuts. That is, 
it constructs groups of size $\Theta(f)$ with a weak diameter of $O(f)$, such that each edge $e\in E$ is associated with a constant number of groups that communicate over it. A formal definition of this concept is as follows.

\begin{definition}[Low-Congestion Shortcuts \cite{ghaffari2016distributed}]
    Given a graph $G=(V,E)$ and a partition of $V$ into disjoint subsets $Q_1,\dots,Q_K\subseteq V$, a set of subgraphs $H_1,\dots H_K\subseteq G$ is called a set of \emph{$(c,d)$-shortcuts}, if:
    \begin{enumerate}
        \item The subgraph $G[Q_i]\cup H_i$ is connected.
        \item For each $i$, the diameter of the subgraph $G[Q_i]\cup H_i$ is at most $d$.
        \item For each edge $e\in E$, the number of subgraphs $G[Q_i]\cup H_i$ containing $e$ is at most $c$.
    \end{enumerate}
\end{definition}

With the above definition, we can now state the properties of our partitioning algorithm.
\begin{theorem}\label{theorem:split to groups}
    Let $G=(V,E)$ be a graph, and let $S$ be a $(2f+2,2f+1)$-ruling set where each node $v$ knows whether it is in $S$. There exists a deterministic \congest algorithm that partitions $V$ into disjoint groups $Q_1,\dots Q_K\subseteq V$ with sizes ranging from $f+1$ to $3f+1$, and provides a set of $(2,4f)$-shortcuts for these groups. The algorithm completes in $\Theta (f)$ rounds.
\end{theorem}

\textbf{Overview.} 
Algorithm \ref{alg: partitioning} constructs the required partition as follows. Initially, each node $v\in S$ computes a BFS tree to depth $2f+1$ by propagating its ID. Nodes that receive multiple root IDs choose the smallest and use that tree only. The rest of the algorithm operates within each tree independently.

Within each tree, each node $u$ divides the nodes in its subtree into groups of size $f+1\leq s\leq 3f+1$, with a remainder of up to $f$ nodes. This remaining number is sent to its parent $parent(u)$, who handles dividing the remainders of its children similarly. The $\ComputeGroups(u)$ procedure sorts $u$'s children by increasing IDs, appends itself, and then forms groups by summing the remaining values of its children. Once a group reaches the desired size, a new group starts from the next child. Each child receives its group ID from $u$.
When a node receives its group allocation, it may need to inform some of its children to ensure that remaining nodes from its subtree join this group as well. This is handled by the $\RelayGroups(u)$ procedure, in which node $u$ relays the group ID. During this process, $u$ also sets the local variable $H_u$ which is used for constructing low-congestion shortcuts for communication among the nodes of each group.
Finally, if there is a small remainder at the root $v$, the root appends these nodes to one of the groups created in $\ComputeGroups(v)$ or, if no such group was constructed by $v$, it appends them to a group with the closest node from this group to $v$. This is done by $v$ obtaining the required group ID and relaying it to the relevant remaining nodes.

Following is the formal algorithm and proof.

\begin{algorithm}[ht!]
\caption{Partitioning algorithm}\label{alg: partitioning}
\setcounter{AlgoLine}{0}
\KwIn{A graph G and a $(2f+2,2f+1)$-ruling set $S$, where each node $v$ is provided with $b(v)$ such that $S=\{v\mid b(v)=1\}$.}
\KwOut{Each node $v$ holds $Group(v)$ and $H_{group}(v)$ for every group $group$.}
\BlankLine
    \nonl\textbf{Constructing BFS trees:}\\
    Every node $v\in S$  starts a BFS to a distance of $2f+1$. Every node $u\notin S$ which receives BFS messages selects the one with the lowest ID and continues with it. Every node $u\notin S$ in a tree rooted at $v$, designates $parent(u)$ as one of the nodes from which it received $v$, and sets $leader(u)=v$. Every node $u$ sends a bit ``1" to $parent(u)$. Let $children(u)$ denote the set of nodes from which $u$ receives such a bit.\label{partstep:bfs}
\BlankLine
    \nonl\textbf{Creating groups:}\\
    Every node $u\in V$ with $children(u)=\emptyset$ sends $remain(u) = 1$ to $parent(u)$. Every node $u\in V$ sorts its children by increasing IDs and appends itself last to create a sequence $order(u)=(w_1,\dots,w_{r_u},u=w_{r_u+1})$, where $r_u=|children(u)|$.\label{partstep:leaf}

    Every node $u\in V$ that receives $remain(w)$ from every node $w\in children(u)$ locally computes $Groups(u)$, $Messages(u)$ and $remain(u)$ using the $\ComputeGroups(u)$ procedure. For each item $(child,group)\in Messages(u)$, $u$ sends $group$ to $w_{child}$. 
    Additionally, if $u\notin S$ it sends $remain(u)$ to $parent(u)$. \label{partstep:ComputeGroups}
\BlankLine
    \nonl\textbf{Relaying group IDs:}\\
    Every node $u\notin S$ that receives $group$ from its parent computes $Messages_r(u)$ using the \RelayGroups$(u,group)$ procedure and sets $Group(u)=group$. For each item $(child,group)\in Messages_r(u)$, $u$ sends $group$ to its child $w_{child}$.\label{partstep:RelayGroups}
\BlankLine
    \nonl\textbf{Root remainder:}\\
    Every node $u\notin S$ sends a group ID $group$ to $parent(u)$. If $Group(u)\ne \emptyset$ then $group$ is $Group(u)$. If $Group(u)=\emptyset$ then $group$ is the first group identifier that $u$ receives from some child (the smallest ID if there is more than one).\label{partstep:UpGroup}

    Every node $v\in S$ with $Group(v)=\emptyset$ sets $Group(v)=group$, where $group$ is the first group identifier that $v$ receives from some child (the smallest ID if there is more than one). Then, $v$ sends $group$ to every child $w$ with $remain(w)>0$ and $Group(w)=\emptyset$ and to the node it received $group$ from in Step \ref{partstep:UpGroup}. For every such node $w$, the node $u$ inserts the edge $(u,w)$ into the set $H_{group}(u)$.\label{partstep v choose}

    Every node $u\notin S$ with $Group(u)=\emptyset$ which receives $group$ from $parent(u)$ sets $Group(u)=group$ and sends $group$ to every child $w$ with $remain(w)>0$ and $Group(w)=\emptyset$ and to the node it received $group$ from in Step \ref{partstep:UpGroup} (if there exists such node). For every such node $w$, the node $u$ inserts the  edge $(u,w)$ into $H_{group}(u)$.\label{partstep u down}
    \BlankLine
    \nonl\textbf{Assigning shortcut edges:}\\
    Denote in $Q_1,\dots,Q_K$ as the groups created by all nodes and define $H_i=\cup_{u\in V}H_{ID(Q_i)}(u)$.
\end{algorithm}

~\\\textbf{Procedure $\ComputeGroups(u)$.}
The node $u\in V$ initializes a counter $count=0$, an empty sequence $Groups(u)=\emptyset$, an empty set $Messages(u)=\emptyset$, and a variable $remain(u)=1$. 
It then loops over $1\leq i\leq r_u+1$ and computes $count \leftarrow count + remain(w_i)$ until $count\geq f+1$.

When this happens for some index $i=j_1$, the node $u$ adds the item $(1, j_1,ID(w_{1}))$ to the sequence $Groups(u)$, sets $count \leftarrow 0$, and continues looping over the nodes in $order(u)$ from $i=j_1+1$ in the same manner. That is, it computes $count \leftarrow count + remain(w_i)$ until $count\geq f+1$ for some value $j_2$, after which it adds the item $(j_1, j_2,ID(w_{j_1}))$ to $Groups(u)$, and so on. Thus, every item in $Groups(u)$ is of the form $(start, end,group)$, where $group$ is the ID of the node $w_{start}$, which serves as the group identifier for all nodes $w_i$ for $start\leq i\leq end$.

Let $(a,b,ID(w_{a}))$ be the last group created by $u$, so $w_b$ is the last child of $u$ whose $remain(w_b)$ value is grouped. 
The node $u$ computes $remain(u)=(\sum_{i=b+1}^{r_u}remain(w_i))+1$.
If $u\in S$ and $Groups(u)\ne \emptyset$, $u$ updates $(a,b,ID(w_a))$ to $(a,r_u+1,ID(w_a))$.

For each item $(start, end, group)$ in $Groups(u)$, $u$ inserts a tuple of $(group,i)$ into $Messages(u)$ for every child $w_i$ where $start \leq i\leq end$ and $remain(w_i)>0$. Thus, every item is of the form: $(child,group)$.
For each item $(child,group)\in Messages(u)$, the node $u$ inserts the edge $(u,w_{child})$ into the set $H_{group}(u)$.

~\\\textbf{Procedure $\RelayGroups(u,group)$.}
Let $(a,b,ID(w_{a}))$ be the last group that $u$ created during its $\ComputeGroups(u)$ procedure. The node $u$ initializes an empty set $Messages_r(u)=\emptyset$. Then, $u$ inserts a tuple $(i,group)$ into $Messages_r(u)$ for every child $w_i$ where $b+1\leq i\leq r_u$. Thus, every item is of the form $(child,group)$.
As in $\ComputeGroups(u)$, for each item $(child,group)\in Messages_r(u)$, the node $u$ inserts the edge $(u,w_{child})$ into the set $H_{group}(u)$.

~\\This concludes the description of the algorithm. 
Appendix \ref{app:proofs} has the proof of Theorem \ref{theorem:split to groups}, as well as for the following claim derived from it, which will be useful in the subsequent section.

\begin{claim}\label{claim: parallel passing info between groups}
In parallel, each node can collect $Y\cdot f$ bits of information from all nodes in its group within $O(f+(Y\cdot f)/\log n)$ rounds, where each node holds $Y$ bits of information and the group sizes are $O(f)$.
\end{claim}

\section{Resilient Labeling Schemes} \label{section: f-resilient labeling schemes}

In this section, we wrap up the previous parts to obtain our resilient labeling scheme.

\ThmMain*

\textbf{Overview.} Algorithm \ref{alg:resilientscheme} gives our full resilient labeling scheme. The oracle executes Algorithm \ref{alg: assigning labels} to assign new labels to each node. In the distributed algorithm, each node parses its new label, which is done correctly since the length of each item in the string is known. The node then uses this label to reconstruct the greedy $(2f+2,2f+1)$-ruling set $S$ (with $f=80\F$) using Algorithm \ref{alg: restoring greedy ruling set} and to obtain its group in the partition according to Algorithm \ref{alg: partitioning}. Finally, the nodes exchange all their labels within a group and decode their original labels, as guaranteed by Theorem \ref{theorem:code analysis}.

\begin{algorithm}[htbp]
\caption{Resilient Labeling Scheme}\label{alg:resilientscheme}
\setcounter{AlgoLine}{0}
\nonl\textbf{Input to Oracle:} A graph $G$, a parameter $\F$, and a labeling $\varphi:V\rightarrow \{0,1\}^\ell$\\
\nonl\textbf{Output of Oracle:} A labeling $\psi:V\rightarrow \{0,1\}^{\ell'}$\\
\nonl\textbf{Input to Distributed Algorithm:} Each node $v$ gets the label $\psi(v)$\\
\nonl\textbf{Output of Distributed Algorithm:} Each node $v$ outputs $\varphi(v)$\\
    \BlankLine
    \nonl\textbf{Oracle:}\\
    Run Algorithm \ref{alg: assigning labels} and assign $\psi(v)$ to each node $v$ \nonl\label{schemestep:oracle}.
    \BlankLine
    \textbf{The distributed algorithm:}\\ 
    For every non-faulty node $v$, parse $\psi(v)$ as $b(v)\circ distS(v) \circ {\tt{w}}(v)$\label{schemestep:parse}.
    
    Run Algorithm \ref{alg: restoring greedy ruling set} with $\zeta(v)=b(v)\circ distS(v)$ for each non-faulty node $v$ to obtain $b(v)$ for every node $v$.\label{schemestep:restore}
    
    Run Algorithm \ref{alg: partitioning} with $b(v)$ for every node $v$, to obtain $Group(v)$ for every node $v$\label{schemestep:partition}.
    
    Every node $v$ sends $(v,{\tt{w}}(v))$ to all nodes in $Group(v)$\label{schemestep:collect code}.
    
    Every node $v$ outputs $decode(v,\{(u,{\tt{w}}(u)) \mid u \in Group(v)\})$\label{schemestep:decode}.
\end{algorithm}

\begin{proof}[Proof of Theorem \ref{thm:main}]

First, note that since $b(v)$ is a single bit and $distS(v)$ has $O(\log{\F})$ bits where the constant is known, the parsing of $\psi(v)$ as $b(v)\circ distS(v) \circ {\tt{w}}(v)$ in Step \ref{schemestep:parse} is unique and matches Step \ref{oraclestep:encode} in Algorithm \ref{alg: assigning labels}.

\textbf{Correctness:} 
By Theorem \ref{theorem:code analysis}, the values $b(v)$ correspond to a greedy $(2f+2,2f+1)$-ruling set $S$, meaning that $\{v \mid b(v) = 1\} =S$. Thus, according to Theorem \ref{theorem:black or white}, after running Algorithm \ref{alg: restoring greedy ruling set} in Step \ref{schemestep:restore}, each node $v$ holds its value $b(v)$. 
Since Algorithm \ref{alg: partitioning} is deterministic and the oracle uses the same algorithm in Algorithm \ref{alg: assigning labels} in Step \ref{schemestep:oracle}, after running Algorithm \ref{alg: partitioning} in Step \ref{schemestep:partition}, each node $v$ will have the same $Group(v)$ as determined by the oracle in Step \ref{oraclestep:partition} of Algorithm \ref{alg: assigning labels}.
Therefore, by Theorem \ref{theorem:code analysis}, the operation $decode(v,\{(u,{\tt{w}}(u)) \mid u \in Group(v)\})$ in Step \ref{schemestep:decode} will yield the value of $\varphi(v)$ for each node $v$.
        
\textbf{Round complexity:} By Theorems \ref{theorem:black or white} and \ref{theorem:split to groups}, Algorithms \ref{alg: restoring greedy ruling set} and \ref{alg: partitioning} in Steps \ref{schemestep:restore} and \ref{schemestep:partition}, respectively, each take $O(f)$ rounds. Additionally, according to Claim \ref{claim: parallel passing info between groups}, every node can collect $O(\ell\cdot f)$ bits of information from all nodes in its group in parallel, within $O(f+(\ell\cdot f)/\log n)$ rounds, as each node holds $O(\ell)$ bits of information and each group is of size $O(f)$. Thus, Step \ref{schemestep:collect code} completes in $O(f+(\ell\cdot f)/\log n)$ rounds.
In total, the distributed algorithm completes in $O(f+(\ell\cdot f)/\log n)$ rounds. Since $f=\Theta(\F)$, we have $O(\F+(\ell\cdot \F)/\log n)$ rounds.

\textbf{Label size:} According to Theorem \ref{theorem:code analysis}, each of the original $\ell$ bits corresponds to a constant number of bits in the coded version. Additionally, there is one bit for the ruling set indication $b(v)$ and $O(\log{\F})$ bits for $distS(v)$, making the new label size $O(\log(\F)+\ell)$. 
\end{proof}

\acknowledgements{This research is supported in part by the Israel Science Foundation
(grant 529/23). We thank Alkida Balliu and Dennis Olivetti for useful discussions about the state of the art for computing ruling sets.}

\bibliography{Bibliography.bib}

\appendix
\section{Missing Proofs}
\label{app:proofs}

\subsection{Proof of Theorem \ref{theorem:code analysis}}
In order to prove this theorem, we need the following claim on $N_M$.
\begin{claim}
\label{claim:N_M}
For every $M>1$, it holds that $N_{M}\leq 5N_{M-1}$.
\end{claim}

\begin{proof}[Proof of Claim \ref{claim:N_M}]
For every $M>0$, the value of $N_M$ is defined as $2M(2^M-1)$, and so we have for $M>1$:
\begin{align*}
N_{M} &=
2M(2^{M}-1)\\
&=2(M-1)\cdot 2^{M}+2\cdot 2^{M}-2(M-1)-2\\
&=2(2(M-1)\cdot 2^{M-1})-2(2(M-1))+2\cdot 2^{M}+2(M-1)-2\\ 
&\leq 2N_{M-1} + 2\cdot 2(M-1)(2^{M-1}) - 2\cdot2(M-1) + 6(M-1) - 2\\
&= 4N_{M-1} + 6(M-1) - 2\\
&\leq 5N_{M-1},
\end{align*}
where the last inequality follows since $M>1$ and thus $6(M-1)-2<N_{M-1}$.
\end{proof}

We now prove the aforementioned properties of the code implied by our choice of parameters using a very crude analysis of the constants, which is sufficient for our needs.

\begin{proof}[Proof of Theorem \ref{theorem:code analysis}]
    Since $M$ is the smallest such that $k_M\geq s$, we have that $k_{M-1}<s$, or in other words, $\rho_{M-1}N_{M-1}<s$. We prove in Claim \ref{claim:N_M} above that $N_{M}$ is at most $5N_{M-1}$. Combining these two with the fact that $\rho_{M-1}\geq \rho=1/4$, gives that $$(1/4)N_{M-1}  \leq \rho_{M-1}N_{M-1} \leq s \leq k_M = \rho_{M} N_{M} \leq \rho_{M}\cdot 5 N_{M-1} \leq 5N_{M-1},$$ where the last inequality follows since $\rho_M \leq 1$. This gives that $k_M/s$ is at most $20$.  
    To encode, we take the $s$ bits and pad them with $k_M-s$ zeros concatenated at their end. We get $N_M$ bits that are split among the $s$ nodes, so now the number of bits that each node has is at most $N_M/s=k_M/(\rho_M\cdot s)$. We know that $\rho_M \geq 1/4$ and that $k_M/s$ is at most $20$, so each node holds at most $80$ bits. This proves Item (1).

    To prove Item (2), note that the code ${\tt{C}}_M$ has a relative distance $\delta_M > (1-2\rho)H^{-1}(1/2) \geq (1-2\cdot (1/4)) \cdot 0.11 \geq 1/20$.
    The code can fix up to $\delta_M N_{M} -1$ erasures, which is more than a $1/40$ fraction of $N_M$. Since $f=80\F$ and the number of blocks $s$ is at least $f+1$ (which is at least $f$), any erasure of $\F$ blocks erases at most a $1/80$ fraction of the blocks. Some blocks may be of size $\ceil{|{\tt{w}}|/s}$ and others of size $\floor{|{\tt{w}}|/s}$, which may differ by 1 and the larger blocks may get erased. However, even if we denote $x=\ceil{|{\tt{w}}|/s}$ and assume $\floor{|{\tt{w}}|/s}=x-1$, then in the worst case, the fraction of bit erasures caused by $\F$ block erasures is at most $(sx/80)/(sx-(79s/80))=x/(80x-79)$, which is at most $1/40$ (since $x>1$). Thus, this may erase at most a $1/40$ fraction of bits of ${\tt{w}}$. Because the code can fix up to a $1/40$ fraction of erasures, we can decode ${\tt{r}}$, obtaining Item (2) as needed. 
\end{proof}

\subsection{Proof of Lemma \ref{lemma: greedy}}
\begin{proof}[Proof of Lemma \ref{lemma: greedy}]
Let $v$ and $u$ be alternative nodes, and let $w\ne v,u$. We claim that $dist(v,w)\leq 2f+1$ if and only if $dist(u,w)\leq 2f+1$. This implies that if $v\in S$ then $(S\setminus\{v\})\cup\{u\}$ is a $(2f+2,2f+1)$-ruling set. In particular, this means that apart from $v$, the node $u$ does not have a node $w\in S$ within distance at most $2f+1$ from it. Therefore, $ID(u)$ cannot be smaller than $ID(v)$, since otherwise the greedy algorithm would reach $u$ before $v$ when iterating over the nodes and would inserted $u$ into $S$.

    To prove that $dist(v,w)\leq 2f+1$ if and only if $dist(u,w)\leq 2f+1$, we assume that $dist(v,w)\leq 2f+1$ and prove that $dist(w,u)\leq 2f+1$. The proof for the other direction is symmetric. Assume towards a contradiction that it is not the case, so $dist(u,w)\geq 2f+2$. Thus, using the triangle inequality we have
    \begin{align*}
        2f+2&\leq
        dist(u,w) 
        \leq dist(u,v)+dist(v,w) 
        \leq f+dist(v,w),
    \end{align*}
    where the last inequality follows from the definition of $u,v$ as alternative nodes. Thus, $dist(v,w)\geq f+2$. Assume towards a contradiction that there is a node $t$ on $P_{v,w}$ such that $dist(v,t)=dist(u,t)$. Since $t$ is on $P_{v,w}$ then $|P_{v,w}|=|P_{v,t}|+|P_{t,w}|$. Therefore, 
    \begin{align*}
        dist(u,w)&\leq |P_{u,t}|+|P_{t,w}|\\
        &=|P_{u,t}|+(|P_{v,w}|-|P_{v,t}|)\\
        &=dist(u,t)+(dist(v,w)-dist(v,t))\\
        &=dist(v,t)+(dist(v,w)-dist(v,t))\\
        &=dist(v,w)\\
        &\leq 2f+1.
    \end{align*}
    This contradicts the assumption that $dist(u,w)\geq 2f+2$. Thus, every node $t$ on $P_{v,w}$ satisfies $dist(u,t)\ne dist(v,t)$.

    Because $dist(v,w)\geq f+2$, there are at least $f+1$ nodes $q$ on $P_{v,w}$ such that $dist(q,u)\ne dist(q,v)$, and at least $f-1$ of them are at a distance of at most $f+1$ from $v$. Thus, $u$ is not an alternative node for $v$, contradicting the assumption. Therefore, $dist(u,w)\leq 2f+1$.
\end{proof}

\subsection{Proof of Theorem \ref{theorem:split to groups}}

\begin{proof}[Proof of Theorem \ref{theorem:split to groups}]
    Since $S$ is a $(2f+2,2f+1)$-ruling set, every node $u\in V$ has a node $v\in S$ within distance of at most $2f+1$ from it. Thus, $u$ belongs to exactly one tree, and in the BFS in Step \ref{partstep:bfs} in the algorithm, the node $u$ chooses one node to be its leader. Let $T_v$ be the tree rooted at $v$.

    ~\\
    \textbf{The BFS tree $T_v$ of every node $v\in S$ satisfies that $|T_v|\geq f+1$:} 
    For every node $v\in S$, there are at least $f+1$ nodes within distance $f$ from it (including $v$ itself), since the graph is connected. These nodes cannot be within distance $f$ from any other node $v'\in S$, and thus they join the BFS tree $T_v$. This proves that $|T_v|\geq f+1$ for all trees. 
    In particular, there exists a partition of the nodes in $T_v$ into groups of sizes at least $f+1$.
    
    ~\\
    \textbf{The algorithm obtains a partitioning of the nodes in every tree into groups with sizes ranging from $f+1$ to $3f+1$:} Fix $v\in S$ and consider only nodes in $T_v$.
    First, we prove that all groups created until the end of Step \ref{partstep:RelayGroups} are of size ranging from $f+1$ to $2f+1$.
    Let $u\in T_v$. In Step \ref{partstep:ComputeGroups}, the node $u$ runs the procedure $\ComputeGroups(u)$.
    During the procedure, when creating a group, the node $u$ starts with the first child in $order(u)$ and adds children according to $order(u)$ until the counter $count$ of their $remain$ values is at least $f+1$. For a group that starts with child $w_{a'}$ and ends with child $w_{b'}$, in Step \ref{partstep:ComputeGroups} the node $u$ sends this group's identifier to every child $w_i$ for which $a'\leq i\leq b'$. Then, each $w_i$ sends in $\RelayGroups(w_i)$ in Step \ref{partstep:RelayGroups} this group identifier to all its children for which it did not assign a group yet during the procedure $\ComputeGroups(w_i)$, and its children relay this identifier in the same manner. Thus, the value of $count$ is the size of the group, and so this size is at least $f+1$.

    Notice that if $count$ is at most $f$, when adding $remain(w')\leq f$ of some child $w'$, we get that the updated $count$ is at most $f+f=2f$ (any $remain$ value is always at most $f$, as otherwise $w'$ would form a group from these at least $f+1$ nodes). A special case is when a root of a subtree $u$ adds itself to a group it created, and then $count$ is at most $2f+1$. Therefore, since the size of a group is the value of $count$, then the size of a group ranges from $f+1$ to $2f+1$.

    This proves that all groups created until the end of Step \ref{partstep:RelayGroups} are of size ranging from $f+1$ to $2f+1$, but there still could be some values that $u$ needs to handle whose sum is $remain(u)$. For every node $u\notin S$, an ancestor of $u$ handles these values. But, the algorithm for the root $v\in S$ in this case differs, as it has no ancestors. Steps \ref{partstep:UpGroup}-\ref{partstep u down} provide the root with a group ID to which it adds its remainder, so the size of that group may become at most $2f+1+f=3f+1$. Thus, all the group sizes range from $f+1$ to $3f+1$.

    Notice that since every node belongs to exactly one group, then all the groups are disjoint.

    ~\\
    \textbf{For each edge $e\in E$, the number of subgraphs $G[Q_i]\cup H_i$ containing $e$ is at most $2$:}
    An edge $(u,w)$ is inserted into $H_{group}(u)$ if $u$ sends the group identifier $group$ to $w$ and this happens only if $w$ sends $remain(w)>0$ to $u=parent(w)$ in Step \ref{partstep:ComputeGroups}. There are two cases. The first is that $u$ assigns $w$ a group during the steps of creating groups. Thus, $u$ does not assign $w$ another group in this or later steps. Therefore, in this case, the edge $(u,w)$ is inserted into only one set $H_{group}(u)$ of $u$. The same argument applies if the edge is inserted into $H_{group}(u)$ during the steps of relaying group IDs.

    The second case is that $u$ receives $group$ from its parent during the steps of handling the root remainder. Then, the root $v$ satisfies that $remain(v)>0$. Thus, according to the algorithm, it assigns additional $remain(v)$ nodes to an existing group. Therefore, $(u,w)$ can be inserted into at most one additional group. Overall, there are at most two groups for which $u$ assigns $e=(u,w)$ to $H_{group}(u)$.

    ~\\
    \textbf{The diameter of the subgraph $G[Q_i]\cup H_i$ is at most $4f$}:
    Consider two cases, based on whether the group contains the root $v$ or not. 
    
    ~\\Case 1: Let $Q$ be a group which does not contain the root $v$, and let $u$ be the node that builds $Q$ during $\ComputeGroups(u)$. Since the group identifier is the ID of the first child in the group, there can be only a single $u$ that creates a group with this ID, and only nodes in the subtree of $u$ can belong to it.  
    Let $t\in Q$. In Step \ref{partstep:RelayGroups}, $parent(t)$ sends $ID(Q)$ to $t$ and inserts the edge $(parent(t),t)$ into $H_{ID(Q)}(parent(t))$. Additionally, if $u\ne parent(t)$, then $parent(t)$ also receives $ID(Q)$ from its parent and sets $Group(parent(t))=Q$ in Step \ref{partstep:RelayGroups} of procedure $\RelayGroups(parent(t))$. This process continues until we reach a node $w$ which is a child of $u$. Since $w$ also receives $ID(Q)$ from $u$ in Step \ref{partstep:ComputeGroups}, then $u$ inserts the edge $(u,w)$ into $H_{ID(Q)}(u)$, and $w$ sets $Group(w)=Q$ in Step \ref{partstep:RelayGroups} of procedure $\RelayGroups(w)$. Thus, on the path from $u$ to $t$, all the edges belong to $H_{ID(Q)}$, and all the nodes on the path from $w$ to $t$ belong to $Q$ including $w$ and $t$. Since this holds for every node $t\in Q$, the subgraph $G[Q]\cup H_{ID(Q)}$ is connected, and the distance between $u$ and $t$ in the tree is bounded by $f$, as otherwise, $w$ has enough nodes to create a group of size $f+1$ which includes $t$, contradicting the fact that $t$ is in the group $Q$ created by $u$. Thus, the diameter of every such group is bounded by $2f$.

    ~\\Case 2: Let $Q$ be the group which contains the root $v$.  
    There are two possibilities. If $Groups(v)\ne \emptyset$, then the remainder of $v$ is assigned by $v$ itself to one of the groups it creates, and the same argument as in the previous case applies.
    
    Otherwise, if $Groups(v)=\emptyset$, then in the steps handling the root remainder, the node $v$ assigns its remainder of nodes to a group that one of its descendants created. We prove that the distance of $v$ from the deepest node in $Q$ is bounded by $2f$. This implies that the diameter of $G[Q]\cup H_i$ is bounded by $4f$.

    Let $t$ be the node which created the group $Q$ in its $\ComputeGroups(t)$ procedure. If $t\in Q$ before handling the root remainder, then $Group(parent(t))=\emptyset$ at that point since otherwise, $v$ would receive $Group(parent(t))$ during Step \ref{partstep:UpGroup} instead of $Group(t)$ and thus $v$ would not be in $Q$ that $t$ created. Thus, $parent(t)$ sets $Group(parent(t))=ID(Q)$ and the same argument applies to every ancestor of $parent(t)$. This implies that the distance from $t$ to the root $v$ is at most $f$, since all of the nodes on that path join $Q$, and if the distance was larger then there would be enough nodes to create a group that is separate from $Q$.
    Moreover, the distance from $t$ to every node in its subtree that is in $Q$ is at most $f$, as otherwise the child of $t$ in that subtree would have $f+1$ to create that group. Together, this means that the distance from $v$ of any node in $Q$ is at most $f+f=2f$.

    If $t\notin Q$ before handling the root remainder, then a similar argument applies: any child $w$ of $t$ which is in $Q$ has distance at most $f-1$ to any other node in its subtree that is in $Q$. The node $t$ itself is again within distance at most $f$ to $v$. Thus, the distance from $v$ of any node in $Q$ is at most $(f-1)+1+f=2f$.
    
    ~\\The above proves that the set of subgraphs $H_1,\dots,H_K$ is a set of $(2,4f)$-shortcuts.

    ~\\\textbf{Round complexity:} Note that for every $v\in S$ and $u\in T_v\setminus \{v\}$, we have that $dist(u,v)\leq 2f+1$. Thus, downcasting and upcasting $f$ messages in $T_v$ takes at most $O(f)$ rounds by standard pipelining. Since we only downcast or upcast a constant number of times, we have that the algorithm completes in $O(f)$ rounds.
\end{proof}

\subsection{Proof of Claim \ref{claim: parallel passing info between groups}}
\begin{proof}[Proof of Claim \ref{claim: parallel passing info between groups}]
    Every node $u\in V$ holds a set $H_{group}(u)$ for every group $group$ that needs to communicate through $u$. This set contain edges to children of $u$. From Theorem \ref{theorem:split to groups}, for each edge $e\in E$, the number of subgraphs $G[Q_i]\cup H_i$ containing $e$ is at most $2$. Thus, every node $u$ can send to its children which edges belong to which group so that the group can communicate over them. Thus, all groups can communicate over their edges in parallel. Additionally, the size of every group $Q_i$ is $O(f)$. Therefore, on every edge, we need to pass $O(f)$ messages, where each message is of size $Y$. Moreover, according to Theorem \ref{theorem:split to groups}, for each group $i$, the diameter of the subgraph $G[Q_i]\cup H_i$ is at most $4f=O(f)$. Hence, by standard pipelining, for every node, it can collect $Y\cdot f$ bits of information from all nodes in its group, within $O(f+{(Y\cdot f)}/{\log n})$ rounds.
\end{proof}

\end{document}